\def\year{2020}\relax
\documentclass[letterpaper]{article} 
\usepackage{aaai20}  
\usepackage{times}  
\usepackage{helvet} 
\usepackage{courier}  
\usepackage[hyphens]{url}  
\usepackage{graphicx} 
\usepackage{graphicx}  
\usepackage{graphicx}
\usepackage{amsthm}
\usepackage{amssymb}
\usepackage{amsmath}
\usepackage{tikz}
\usepackage{algpseudocode}
\usepackage{algorithm}

\usetikzlibrary{arrows,positioning,shapes,calc,decorations.shapes, decorations.pathreplacing}
\newtheorem{definition}{Definition}
\newtheorem{theorem}{Theorem}
\newtheorem{lemma}{Lemma}
\newtheorem{proposition}{Proposition}

\newtheorem{example}{Example}

\newcommand{\black}[1]{\textcolor{black}{#1}}

\title{Strategy-Proof and Non-Wasteful Multi-Unit Auction \\via Social Network}
\author{Takehiro Kawasaki\textsuperscript{\rm 1}, Nathana\"{e}l Barrot\textsuperscript{\rm 2}, Seiji Takanashi\textsuperscript{\rm 3}, Taiki Todo\textsuperscript{\rm 1,2}, \and Makoto Yokoo\textsuperscript{\rm 1,2}\\
\textsuperscript{\rm 1}Kyushu University, Japan, \{kawasaki@agent., todo@, yokoo@\}inf.kyushu-u.ac.jp\\
\textsuperscript{\rm 2}RIKEN AIP, Japan, nathanaelbarrot@gmail.com\\
\textsuperscript{\rm 3}Kyoto University, Japan, s.takanashi1990@gmail.com\\ 
}
\begin{document}

\maketitle

\begin{abstract}
Auctions via social network, pioneered by Li et al.\ (2017),
have been attracting considerable attention in the literature
of mechanism design for auctions.
However, 
no known mechanism has satisfied strategy-proofness,
non-deficit, non-wastefulness, and individual rationality 
for the multi-unit unit-demand auction, except for some na\"{i}ve ones.
In this paper, we first propose a mechanism that satisfies all the
above properties.
We then make a comprehensive comparison with two na\"{i}ve mechanisms,
showing that the proposed mechanism dominates them in social surplus,
seller's revenue, and incentive of buyers for truth-telling.
We also analyze the characteristics of the social surplus and the revenue
achieved by the proposed mechanism, including the constant approximability of
the worst-case efficiency loss and the complexity of optimizing revenue
from the seller's perspective.
\end{abstract}

\section{Introduction}
\label{sec:intro}

Auction theory has attracted much attention in
artificial intelligence
as a foundation of multi-agent resource allocation.
One of the mainstreams in the literature is
analyzing auctions from the perspective of
mechanism design.
In particular, several works studied
how to design {\em strategy-proof} auctions,
which incentivize each buyer to truthfully report
her valuation function, regardless of the reports of the other buyers.
One critical contribution in the literature
is the development of the Vickrey-Clarke-Groves mechanism (VCG),
which satisfies strategy-proofness and
various other properties~\cite{vickrey:JOF:1961,clarke,groves:73}.

Li et al.\ (\citeyear{li:AAAI:2017}) proposed a new model of auctions,
in which buyers are distributed in a social network and
the information on the auction propagates over it.
Utilizing a social network, the seller can advertise the auction
to more potential buyers beyond her followers, as many works studied in
network science~\cite{emek:EC:2011,borgatti:SCI:2009,jackson:SEN:2008,kempe:KDD:2003}.
From the buyers' perspective, however,
forwarding the information increases the number of buyers,
which reduces the possibility that they will get the item.
Therefore,
the main challenge in the auction via social network is
how to incentivize buyers to forward the information
to as many followers as possible,
as well as truthfully reporting their valuation functions.
For selling a single unit of an item,
Li et al. (\citeyear{li:AAAI:2017}) developed an auction mechanism in which
each buyer is incentivized to forward the information to her followers.

Zhao et al.\ (\citeyear{zhao:AAMAS:2018}) studied
a multi-unit unit-demand auction via social network,
where each unit is identical and each buyer requires a unit.
They proposed the generalized information diffusion mechanism (GIDM)
and argued that it is strategy-proof.
However, Takanashi et al.\ (\citeyear{takanashi:arXiv:2019})
pointed out an error in their proof
and argued that GIDM is not strategy-proof.
They also proposed a strategy-proof mechanism for the same model,
which however violates a revenue condition called {\em non-deficit},
i.e., the seller might suffer a deficit.
To the best of our knowledge,
for the multi-unit unit-demand auction via social network,
no mechanism satisfying both strategy-proofness and non-deficit
has been developed, except for some na\"{i}ve ones.

The main objective of this paper is
to propose a mechanism that satisfies both strategy-proofness
and non-deficit, as well as some other properties.
As Takanashi et al.\ (\citeyear{takanashi:arXiv:2019}) pointed out,
no mechanism satisfies those properties and Pareto efficiency,
i.e., maximizing the social surplus, under certain natural assumptions.
They thus considered weakening the non-deficit condition.
In this paper, on the other hand,
we consider a weaker efficiency property
called {\em non-wastefulness},
which only requires the allocation of as many units as possible.
Non-wastefulness has its own importance in practice.
For example, in a spectrum auction,
it is important to allocate as much frequency range as possible
to carriers in order to guarantee a sufficient number of services.

We propose a new mechanism, called distance-based mechanism,
for a multi-unit unit-demand auction via social network,
which satisfies strategy-proofness, non-deficit,
non-wastefulness, and individual rationality,
i.e., no buyer receives negative utility under truth-telling,
and whose description is much simpler than GIDM.
It is inspired by the concept of
the {\em diffusion critical tree}, originally proposed in Li et al.\ (\citeyear{li:AAAI:2017}),
which specifies, for each buyer $i$, the set of {\em critical} buyers
for $i$'s participation.
If a buyer $j$ is critical for another buyer $i$'s participation,
i.e., if $i$ cannot participate in the auction without $j$'s
forwarding of information, $j$ must receive a higher priority in
the competition.

We then make a comprehensive comparison with two na\"{i}ve mechanisms
that also satisfy (most of) the above properties.
One is based on VCG, being applied only to the buyers
who are directly connected to the seller. The other mechanism
simply allocates the units in the first-come-first-served manner
with no payment.
We show that the distance-based mechanism dominates
both of these na\"{i}ve ones in terms of
social surplus and the seller's revenue.
Furthermore,
in those mechanisms, hiding the information, combined with reporting the
true value, is also a dominant strategy,
while this is not the case in our mechanism when $k \geq 2$.
This indicates that each buyer has a stronger incentive for truth-telling
in the distance-based mechanism.

We further analyze the characteristics of the social surplus and
the revenue of the distance-based mechanism.
About social surplus, it guarantees that
each winner is in the set of top-$k$ buyers except for her followers.
It also has a constant worst-case efficiency loss,
based on a measure proposed by Nath and Sandholm (\citeyear{nath:GEB:2018}),
when an optimal reserve price is introduced.
About revenue,
we show that a revenue monotonicity condition fails and
that maximizing revenue by
optimally sending the information is NP-complete.

\section{Preliminaries}
\label{sec:prel}

We first define the standard notations for multi-unit unit-demand auctions.
Let $s$ be a seller who is willing to sell
the set $K$ of $k$ identical units.
Let $N$ be the set of $n$ buyers,
where each buyer $i \in N$ has a unit-demand valuation function for $K$.
Let \black{$\boldsymbol{x} = (x_{i})_{i \in N} \subset \{0,1\}^{n}$} be an
allocation, which specifies who obtains a unit, where
$\black{x_{i}} = 1$ indicates that buyer $i$ obtains a unit under allocation
\black{$\boldsymbol{x}$}, and $\black{x_{i}} = 0$ otherwise.
Let $v_{i} \in \mathbb{R}_{\geq 0}$ indicate the true unit-demand {\em value} of buyer $i$
for a single unit.
We assume that each buyer's utility is {\em quasi-linear},
i.e., the utility of buyer $i$ under allocation $\boldsymbol{x}$,
when she pays $p_{i} \in \mathbb{R}$, is given as $v_{i} \cdot \black{x_{i}} -
p_{i}$.

Next, we define additional notations for the auction via social network.
For each buyer $i \in N$,
let $r_{i} \subseteq N \setminus \{i\}$ be the set of buyers to whom
buyer $i$ can forward the information, called $i$'s {\em followers}.
Also, let $r_{s} \subseteq N $ be the set of {\em direct buyers},
i.e., those to whom the seller $s$ can directly send the information.
Given \black{$(r_{i})_{i \in N \cup \{s\}}$},
we define the auction network as a digraph $G = (N \cup \{s\}, E)$,
where for each \black{$i \in N \cup \{s\}$} and each \black{$j \in r_{i}$},
a directed edge, from \black{$i$ to $j$}, is added to the set $E$. 
Note that $r_{i}$ is also private information of buyer $i$ in our model,
so the auction network is defined according to {\em reported} $\boldsymbol{r}' = (r'_{i})_{i \in N}$,
where $r'_{i}$ indicates the set of $i$'s followers to whom $i$
forwards the information.
To summarize, for each $i$,
the private information is given as $\theta_{i} = (v_{i}, r_{i})$,
called the true type of $i$,
consisting of the true value $v_{i}$ and the set $r_{i}$ of the followers.
Any reportable type $\theta'_{i} = (v'_{i}, r'_{i})$ of
$i$ with true type $\theta_{i} = (v_{i}, r_{i})$
satisfies $r'_{i} \subseteq r_{i}$,
i.e., a buyer can only forward the information to her followers.
Let $R(\theta_{i})$ be the set of all reportable types by $i$ with $\theta_{i}$.
Also, let $\boldsymbol{\theta}'$ denote the profile of types reported by all buyers
and $\boldsymbol{\Theta}$ denote the set of all possible type profiles.

For notation simplicity, we introduce additional technical terms
regarding the auction network.
A buyer $i$ is {\em connected}
if a path $s \rightarrow \cdots \rightarrow i$
in $G$ is formed based on the reported $\boldsymbol{r}'$.
Let $\hat{N}$ denote the set of connected buyers.
For each $i$, let $d(i)$ denote the distance of
the shortest path from $s$ to $i$. If $i$ is not connected,
we assume $d(i) = \infty$.
Given $\boldsymbol{\theta}'$,
a buyer $j \in \hat{N}$ is a {\em critical parent of} $i \in \hat{N}$
if, without $j$'s participation, $i$ is not connected,
i.e., $j$ appears in {\em any} path from $s$ to $i$
in $G$. 
Let $P_{i}(\boldsymbol{\theta}') \subseteq \hat{N}$ denote
the set of all critical parents of $i$ under $\boldsymbol{\theta}'$.
The buyer $j \in P_{i}(\boldsymbol{\theta}')$
closest to $i$
is called the {\em least critical parent of} $i$.
An allocation \black{$\boldsymbol{x}$} is {\em feasible}
if $\sum_{i \in N} \black{x_{i}} \leq k$, and $\black{x_{i}} = 1$ implies $i \in
\hat{N}$ for each $i \in N$.
Let $\black{\boldsymbol{X}}$ be the set of all feasible allocations.

Now we are ready to give a formal description of (direct revelation)
mechanisms\footnote{As discussed in the appendix, 
the revelation principle holds under our assumption $R$
on reportable subset $r'_{i} \subseteq r_{i}$.
Therefore, focusing on direct revelation mechanisms is w.l.o.g.}.
A mechanism $(f,t)$ consists of two components,
an allocation rule $f$ and a profile of transfer rules $(t_{i})_{i \in N}$.
An allocation rule $f$ maps a profile $\boldsymbol{\theta}'$ of reported types
to a feasible allocation $f(\boldsymbol{\theta}') \in \black{\boldsymbol{X}}$.
We sometimes use the notation of $f(\theta_{i}', \boldsymbol{\theta}'_{-i})$ instead,
especially when we focus on the report of a specific buyer $i$,
where $\boldsymbol{\theta}'_{-i}$ indicates the profile of types reported by the others.
Given
$\boldsymbol{\theta}'$,
$f_{i}(\boldsymbol{\theta}') \in \{0,1\}$ denotes
the assignment to buyer $i$.
Each transfer rule $t_{i}$ maps
a profile $\boldsymbol{\theta}'$
to a real number $t_{i}(\boldsymbol{\theta}') \in \mathbb{R}$,
which indicates the amount that buyer $i$ pays to the seller.

Here, we define several properties that mechanisms should satisfy.
Feasibility requires that
for any input, the allocation returned by the mechanism is feasible.

\begin{definition}
 A mechanism $(f,t)$ is {\em feasible} if for any $\boldsymbol{\theta}'$,
 $f(\boldsymbol{\theta}')$ is feasible.
\end{definition}

Strategy-proofness is an incentive property,
requiring that, for any buyer,
reporting its true valuation
and forwarding the information to all of its followers
is a dominant strategy. 

\begin{definition}
 Given a mechanism $(f,t)$
 and a buyer $i$ with true type $\theta_{i} = (v_{i}, r_{i})$,
 a report $\theta^{*}_{i} = (v^{*}_{i}, r^{*}_{i}) \in R(\theta_{i})$ is a {\em dominant strategy}
 if
 for
 any $\boldsymbol{\theta}'_{-i}$ and
 $\theta'_{i} \in R(\theta_{i})$,
 \[
 \textstyle
 v_{i} \cdot f(\theta^{*}_{i}, \boldsymbol{\theta}'_{-i}) - t_{i}(\theta^{*}_{i}, \boldsymbol{\theta}'_{-i})
 \geq
 v_{i} \cdot f(\theta'_{i},\boldsymbol{\theta}'_{-i}) - t_{i}(\theta'_{i}, \boldsymbol{\theta}'_{-i})
 \]
 holds.
 A mechanism $(f,t)$ is {\em strategy-proof}
 if reporting $\theta_{i}$ is a dominant strategy for any $i$ under $(f,t)$.
\end{definition}

Individual rationality is a property related
to the incentives of the buyers for participation,
which requires that
truth-telling guarantees a non-negative utility.

\begin{definition}
 A mechanism $(f,t)$ is {\em individually rational} if
 for
 any $i$,
 $\theta_{i}$, and
 $\boldsymbol{\theta}'_{-i}$,
 $v_{i} \cdot
 f(\theta_{i}, \boldsymbol{\theta}'_{-i}) - t_{i}(\theta_{i}, \boldsymbol{\theta}'_{-i})
 \geq 0$
 holds.
\end{definition}

Non-deficit is a property about seller's revenue,
which requires that
the seller's revenue cannot be negative.
Note that it does not consider
each individual transfer and thus does not
imply the non-negativity of each buyer's payment.

\begin{definition}
 A mechanism $(f,t)$ satisfies {\em non-deficit} if
 for
 any $\boldsymbol{\theta}'$,
 $\sum_{i \in N} t_{i}(\boldsymbol{\theta}') \geq 0$
 holds.
\end{definition}

Non-wastefulness is a property about the efficiency of allocation,
which requires that the mechanism allocate as many units as possible.
Note that the traditional definition of non-wastefulness
ignores the network structure,
and thus the second term in RHS
is replaced with $|N|$.

\begin{definition}
 A mechanism $(f,t)$ is {\em non-wasteful} if
 for any $\boldsymbol{\theta}'$,
 $\sum_{i \in N} f_{i}(\boldsymbol{\theta}') \geq \min\{k, |\hat{N}|\}$
 holds.
\end{definition}

\subsection{Two Na\"{i}ve Mechanisms for Comparison}

One might expect that those properties hold in na\"{i}ve mechanisms.
Indeed, we can easily find the following two candidates.
The formal definitions are in the appendix. 
We compare their performances with
that of our new mechanism in the following sections.

The first mechanism applies VCG to only the direct buyers $r_{s}$.
It satisfies strategy-proofness, individual rationality, non-deficit,
and non-wastefulness for $|r_{s}| \geq k$.
We refer to this mechanism as {\em No-Diffusion-VCG} (ND-VCG in
short).
Such a mechanism is also considered in Li et al.\ (\citeyear{li:AAAI:2017}),
although they focused on single-item auctions.

The second mechanism gives the units to buyers for free,
in the first-come-first-served manner,
which is referred to as {\em FCFS-F}.
It satisfies individual rationality, non-deficit, and non-wastefulness,
and it is strategy-proof when
earlier arrivals are not allowed,
e.g., based on ascending order of $d(\cdot)$,
as usually assumed in online mechanism design~\cite{hajiaghayi:EC:2004,todo:AAMAS:2012}.

\section{Distance-Based Mechanism}
\label{sec:mechanism}

The definition of the new mechanism
is given in Definition~\ref{def:proposed-mechanism}.
A key concept in describing the mechanism is
the {\em diffusion critical tree} $T(\boldsymbol{\theta}')$,
originally introduced in Zhao et al.\ (\citeyear{zhao:AAMAS:2018}).
Given $\boldsymbol{\theta}'$, the diffusion critical tree $T(\boldsymbol{\theta}')$
is a rooted tree, where $s$ is the root,
the nodes of $T(\boldsymbol{\theta}')$ are all the connected buyers $\hat{N}$,
and for each node $i \in \hat{N}$ and its least critical parent $j \in P(\boldsymbol{\theta}')$,
an edge $(j,i)$ is drawn.
If $P_{i}(\boldsymbol{\theta}') = \emptyset$, we draw an edge $(s,i)$.
Furthermore, given $T(\boldsymbol{\theta}')$ and a node $i$,
all of the nodes in the subtree of $T(\boldsymbol{\theta}')$ rooted at $i$
are called $i$'s {\em descendants}.
Also, given a report $\boldsymbol{\theta}'$,
a subset $S \subseteq \hat{N}$, and an integer $k' \leq k$,
$v^{*}(S, k')$ denotes the $k'$-th highest value in $S$
under $\boldsymbol{\theta}'$.
For $k' \leq 0$, let $v^{*}(S, k') = \infty$.
In addition, if $|S| < k'$, then $v^{*}(S, k') = 0$.

\begin{definition}
 \label{def:proposed-mechanism}
 Given $\boldsymbol{\theta}'$,
 first order the connected buyers
 $\hat{N}$ in ascending order of $d(\cdot)$,
 with arbitrary fixed tie-breaking.
 Note that $d(\cdot)$ is the distance from $s$ in the original
 graph, not the distance in $T(\boldsymbol{\theta}')$.
 The order $\succ$ is called the {\em priority order}.
 For each $i \not \in \hat{N}$,
 $f_{i}(\boldsymbol{\theta}') = t_{i}(\boldsymbol{\theta}') = 0$.
 For each $i \in \hat{N}$,
 let $\hat{N}_{-i}$ be the set of all connected buyers
 except $i$ and its descendants
 in $T(\boldsymbol{\theta}')$.
 It then runs as follows:
 \begin{algorithmic}[1]
  \State{$k' \leftarrow k, W \leftarrow \emptyset$}
  \For {each $i \in \hat{N}$ selected in the order of $\succ$}
  \State {$p_{i} \leftarrow v^{*}(\hat{N}_{-i} \setminus W, k')$}
  \If {$v'_{i} \geq p_{i}$}
  \State \hspace*{\algorithmicindent} {$f_{i}(\boldsymbol{\theta}') \leftarrow 1, t_{i}(\boldsymbol{\theta}') \leftarrow p_{i}$}
  \State{$k' \leftarrow k'-1, W \leftarrow W \cup \{i\}$}
  \Else
  \State{$f_{i}(\boldsymbol{\theta}') \leftarrow 0, t_{i}(\boldsymbol{\theta}') \leftarrow 0$}
  \EndIf
  \EndFor
 \end{algorithmic}
\end{definition}

The following example demonstrates
how the distance-based mechanism works.
The network of buyers, defined by $\boldsymbol{r}$,
as well as their true values,
is shown in Fig.~\ref{fig:ex}.

\begin{example}
 \label{ex:ex}
 Consider three units
 and
 seven buyers $N = \{i_{1}, i_{2}, \ldots, i_{7}\}$.
 Each vertex in the left figure of Fig.~\ref{fig:ex} corresponds to a buyer,
 and the number in each vertex denotes her true valuation.
 The priority order
 is
 given as $i_{1} \succ i_{2} \succ \cdots \succ i_{7}$.
 Assume that every buyer forwards the information
 to all of her followers,
 i.e.,
 $\hat{N} = N$.
 The corresponding diffusion critical tree
 is given on the right in Fig.~\ref{fig:ex}.

 The assignment to buyers is computed one-by-one,
 in the priority order.
 For buyer $i_{1}$, the price is given as
 $p_{i_{1}} = v^{*}(\hat{N}_{-i_{1}} \setminus W, k-|W|)
 = v^{*}(\{i_{2},i_{3},i_{4},i_{5},i_{6},i_{7}\}, 3) = v_{i_{5}} = 50$.
 Since $p_{i_{1}} > v_{i_{1}}$,
 she does not win a unit.
 For buyer $i_{2}$, the price is given as
 $p_{i_{2}} = v^{*}(\{i_{1},i_{3},i_{5},i_{6},i_{7}\}, 3) = v_{i_{7}} = 40$.
 Note that
 $v_{i_{4}}$
 is ignored because $i_{2}$ is a critical parent of $i_{4}$.
 Since $v_{i_{2}} > p_{i_{2}}$,
 she wins a unit;
 $W$ is updated to $\{i_{2}\}$, and $k$ is decremented to $2$.
 For $i_{3}$, the price is given as $p_{i_{3}} = v^*(\{i_{1},i_{4}\}, 2) = v_{i_{1}} = 30$.
 Since $v_{i_{3}} > p_{i_{3}}$, she wins a unit;
 $W$ is updated to $\{i_{2}, i_{3}\}$, and $k$ is decremented to $1$.
 For $i_{4}$, the price is given as $p_{i_{4}} = v^*(\{i_{1},i_{5},i_{6},i_{7}\}, 1) = v_{i_{6}} = 66$.
 Since $p_{i_{4}} > v_{i_{4}}$,
 she does not win a unit.

 For $i_{5}$, the price is given as $p_{i_{5}} = v^*(\{i_{1},i_{4}\}, 1) = v_{i_{4}} = 45$.
 Since $v_{i_{5}} > p_{i_{5}}$,
 she wins a unit;
 $W$ is updated to $\{i_{2}, i_{3}, i_{5}\}$ and
 $k$ is decremented to $0$.
 Since no unit remains, the prices for the remaining buyers, $i_{6}$ and $i_{7}$,
 become infinity,
 and thus neither of the buyers wins a unit.
 To sum up, $i_{2}$, $i_{3}$, and $i_{5}$ are winners,
 who pay 40, 30, and 45 respectively.
\end{example}

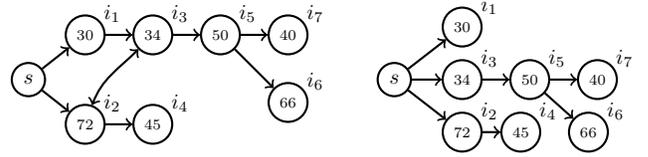
\begin{figure}[t]
\begin{center}
	\begin{tikzpicture}[scale=0.8,transform shape]
	\node[draw, thick, circle] (s) {$s$};
	\node[draw, thick, circle,  label={[xshift = 0.45cm, yshift=-0.25cm] $i_{1}$}] (1) [above right = 0.3cm and 0.5cm of s] {\scriptsize $30$};
	\node[draw, thick, circle,  label={[xshift = 0.45cm, yshift=-0.25cm] $i_{2}$}] (2) [below right = 0.3cm and 0.5cm of s] {\scriptsize $72$};
	\node[draw, thick, circle,  label={[xshift = 0.45cm, yshift=-0.25cm] $i_{3}$}] (3) [right = 0.45cm of 1] {\scriptsize $34$};
	\node[draw, thick, circle,  label={[xshift = 0.45cm, yshift=-0.25cm] $i_{4}$}] (4) [right = 0.45cm of 2] {\scriptsize $45$};
	\node[draw, thick, circle,  label={[xshift = 0.45cm, yshift=-0.25cm] $i_{5}$}] (5) [right = 0.45cm of 3] {\scriptsize $50$};
	\node[draw, thick, circle,  label={[xshift = 0.45cm, yshift=-0.25cm] $i_{7}$}] (7) [right = 0.45cm of 5] {\scriptsize $40$};
	\node[draw, thick, circle,  label={[xshift = 0.45cm, yshift=-0.25cm] $i_{6}$}] (6) [below = 0.45cm of 7] {\scriptsize $66$};
	\node[draw, thick, circle] (ss) [right = 5.5cm of s] {$s$};
	\node[draw, thick, circle,  label={[xshift = 0.45cm, yshift=-0.25cm] $i_{3}$}] (s3) [right = 0.5cm of ss] {\scriptsize $34$};
	\node[draw, thick, circle,  label={[xshift = 0.45cm, yshift=-0.25cm] $i_{1}$}] (s1) [above = 0.2cm of s3] {\scriptsize $30$};
	\node[draw, thick, circle,  label={[xshift = 0.45cm, yshift=-0.25cm] $i_{2}$}] (s2) [below = 0.2cm of s3] {\scriptsize $72$};
	\node[draw, thick, circle,  label={[xshift = 0.45cm, yshift=-0.25cm] $i_{4}$}] (s4) [right = 0.3cm of s2] {\scriptsize $45$};
	\node[draw, thick, circle,  label={[xshift = 0.45cm, yshift=-0.25cm] $i_{5}$}] (s5) [right = 0.45cm of s3] {\scriptsize $50$};
	\node[draw, thick, circle,  label={[xshift = 0.45cm, yshift=-0.25cm] $i_{6}$}] (s6) [right = 0.45cm of s4] {\scriptsize $66$};
	\node[draw, thick, circle,  label={[xshift = 0.45cm, yshift=-0.25cm] $i_{7}$}] (s7) [right = 0.45cm of s5] {\scriptsize $40$};

	\path[draw,thick]
	(s) edge[->]  (1)
	(s) edge[->]  (2)
	(2) edge[<->,out=70, in=225]  (3)
	(1) edge[->]  (3)
	(3) edge[->]  (5)
	(5) edge[->]  (7)
	(5) edge[->]  (6)
	(2) edge[->]  (4);

	\path[draw,thick]
	(ss) edge[->]  (s1)
	(ss) edge[->]  (s2)
	(ss) edge[->]  (s3)
	(s3) edge[->]  (s5)
	(s5) edge[->]  (s7)
	(s5) edge[->]  (s6)
	(s2) edge[->]  (s4);
\end{tikzpicture}
\end{center}
\caption{Example of Buyers Network and Corresponding Diffusion Critical Tree.}
\label{fig:ex}
\end{figure}

Let us clarify how it differs from
GIDM by Zhao et al.\ (\citeyear{zhao:AAMAS:2018})
and how it maintains strategy-proofness.
GIDM first assigns, according to the reported $\boldsymbol{\theta}'$,
a certain number of units to each
subtree of $T(\boldsymbol{\theta}')$.
The buyers in a subtree then compete with each other
to buy the units assigned to it. This is something like
creating a sub-market for each subtree.
However, by not forwarding the information,
some buyer, who originally loses
due to the existence of some winning parent,
can reduce the number of units assigned to the subtree,
make the sub-market more competitive and the parent losing,
and obtain a chance to win.
This is actually the case
found by Takanashi et al.\ (\citeyear{takanashi:arXiv:2019}).

The distance-based mechanism also uses the diffusion critical tree.
However, it does not create such a sub-market for each subtree.
Instead, it has a single market with all of the units,
where buyers' priorities are defined
based on the distance $d(\cdot)$, which is not successfully manipulable;
no buyer can make the distance shorter
by not forwarding the information to her followers,
which is 
shown by Lemma 1.

\subsection{Properties of Distance-Based Mechanism}

We show feasibility, individual rationality, and non-deficit in Theorem~\ref{thm:FIRND},
non-wastefulness in Theorem~\ref{thm:NW},
and strategy-proofness in Theorem~\ref{thm:SP}.
The proofs of Theorems~\ref{thm:FIRND} and \ref{thm:NW} are in the
appendix.
Let $\hat{W}$ denote a set of winners $\{w_1, w_2, \ldots\}$
and $\hat{W}_{\succ j}$ denote $\{w \in \hat{W} \mid w \succ j\}$.

\begin{theorem}
 \label{thm:FIRND}
 The distance-based mechanism satisfies feasibility, individual rationality,
 and non-deficit.
\end{theorem}

\begin{theorem}
 \label{thm:NW}
 The distance-based mechanism
 is non-wasteful.
\end{theorem}

\begin{theorem}
 \label{thm:SP}
 The distance-based mechanism
 is strategy-proof.
\end{theorem}

\begin{proof}
 Let $(f,t)$ be the distance-based mechanism.
 It suffices to show that
 (I) a buyer has no incentive not to forward information to her
 followers, and that
 (II) a buyer cannot obtain any gain by misreporting her value.
 That is, for any $\theta_{i} = (v_{i}, r_{i})$
 and $\theta'_{i} = (v'_{i}, r'_{i}) \in R(\theta_{i})$,
 consider an intermediate type $\theta_{i}^{m} = (v_{i}, r'_{i})$.
 The strategy-proofness condition thus holds from
 (I)
 $
 v_{i} \cdot f_{i}(\theta_{i}, \boldsymbol{\theta}'_{-i}) - t_{i}(\theta_{i}, \boldsymbol{\theta}'_{-i})
 \geq
 v_{i} \cdot f_{i}(\theta^{m}_{i}, \boldsymbol{\theta}'_{-i}) - t_{i}(\theta^{m}_{i}, \boldsymbol{\theta}'_{-i})
 $
 and
 (II)
 $
 v_{i} \cdot f_{i}(\theta^{m}_{i}, \boldsymbol{\theta}'_{-i}) - t_{i}(\theta^{m}_{i}, \boldsymbol{\theta}'_{-i})
 \geq
 v_{i} \cdot f_{i}(\theta'_{i}, \boldsymbol{\theta}'_{-i}) - t_{i}(\theta'_{i}, \boldsymbol{\theta}'_{-i})
 $.
 These inequalities are proven in Lemmas~\ref{lem:r_i} and \ref{lem:v_i}.
\end{proof}

\begin{lemma}
 \label{lem:r_i}
 For any $i$,
 $\theta_{i} = (v_{i}, r_{i})$,
 $\boldsymbol{\theta}'_{-i}$, and
 $\theta'_{i} = (v_{i}, r'_{i})$ s.t.\ $r'_{i} \subset r_{i}$,
 $v_{i} \cdot f_{i}(\theta_{i}, \boldsymbol{\theta}'_{-i}) - t_{i}(\theta_{i}, \boldsymbol{\theta}'_{-i})
 \geq
 v_{i} \cdot f_{i}(\theta'_{i}, \boldsymbol{\theta}'_{-i}) - t_{i}(\theta'_{i}, \boldsymbol{\theta}'_{-i})$
 holds.
\end{lemma}

\begin{proof}
 By not forwarding the information, $i$ can affect another buyer $j$
 in one of the following ways:
 (i) buyer $j$, who is originally a descendant of $i$ in $T(\boldsymbol{\theta}')$, becomes
 disconnected,
 (ii) for buyer $j$, which originally satisfies $i \succ j$,
 the distance $d(j)$ becomes larger.
 In case (i), $j$ is originally not included in $\hat{N}_{-i}$.
 Furthermore, making $j$ disconnected might decrease the
 price of other buyers $j'$ s.t.\ $j' \succ i$. Then
 there is a chance that $i$'s price increases. Thus, not forwarding
 the information is useless in case (i).
 In case (ii), even when $d(j)$ becomes larger,
 $i \succ j$ holds originally, and $i$'s price does not change.
 Thus,
 not forwarding the information is futile.
\end{proof}

\begin{lemma}
 \label{lem:v_i}
 For any $i$,
 $\theta_{i} = (v_{i}, r_{i})$,
 $\boldsymbol{\theta}'_{-i}$, and
 $\theta'_{i} = (v'_{i}, r_{i})$,
 $v_{i} \cdot f_{i}(\theta_{i}, \boldsymbol{\theta}'_{-i}) - t_{i}(\theta_{i}, \boldsymbol{\theta}'_{-i})
 \geq
 v_{i} \cdot f_{i}(\theta'_{i}, \boldsymbol{\theta}'_{-i}) - t_{i}(\theta'_{i}, \boldsymbol{\theta}'_{-i})$
 holds.
\end{lemma}

\begin{proof}
 For buyer $i$, her price  $p_i$ is given as:
 $v^*(\hat{N}_{-i} \setminus \hat{W}_{\succ i}, k - |\hat{W}_{\succ i}|)$.
 It is clear that $p_i \geq v^*(\hat{N}_{-i}, k)$ holds.
 Let $\pi_i$ denote $v^*(\hat{N}_{-i}, k)$.
 $\pi_i$ is determined independently from $i$'s
 declared evaluation value. If $v_i \leq \pi_i$ holds,
 $i$ cannot gain a positive utility regardless of her declaration.
 Thus, assume $v_i > \pi_i$ holds.
 Her actual price, i.e.,
 $p_i = v^*(\hat{N}_{-i} \setminus \hat{W}_{\succ i}, k -
 |\hat{W}_{\succ i}|)$,
 can be strictly larger than $\pi_i$, if
 some buyer $j$ (where $j \succ i$) s.t.\
 $v'_j \leq \pi_i$ becomes a winner.
 Note that if $v'_j > \pi_i$ holds,
 $j$ is within the top $k-1$ winners in $\hat{N}_{-i}$;
 the fact that $j$ becomes a winner does not change $p_{i}$. 

 The only way for $i$ to decrease her price is
 to turn such a winner into a loser by over-bidding.
 Assume $j$ (where $j \succ i$) is such a winner.
 If $j$ is $i$'s ancestor, $i$ cannot affect $j$'s price.
 Thus, $j$ and $i$ are in different branches in $T(\boldsymbol{\theta}')$.
 Since $j$ is a winner,
 $v'_j \geq v^*(\hat{N}_{-j} \setminus \hat{W}_{\succ j}, k -
 |\hat{W}_{\succ j}|)$ holds. Also,
 to increase $j$'s price, $v_i$ must be smaller than or
 equal to $v^*(\hat{N}_{-j} \setminus \hat{W}_{\succ j}, k -
 |\hat{W}_{\succ j}|)$. Note that $i$ is included in
 $\hat{N}_{-j} \setminus \hat{W}_{\succ j}$. If $i$ is
 within the top $k - |\hat{W}_{\succ j}| - 1$ buyers in
 $\hat{N}_{-j} \setminus \hat{W}_{\succ j}$, even if $i$ over-bids,
 she cannot change $j$'s price.
 Thus, $v'_j \geq v_i$ holds.
 However, we assume $v'_j \leq \pi_i < v_i$ holds.
 This is a contradiction. Thus, $i$ cannot decrease her price by
 misreporting her evaluation value.
\end{proof}

\section{Efficiency Analysis}
\label{sec:efficiency}

In this section we conduct a more detailed analysis on efficiency.
We show that any winner has a value that is
in the set of top-$k$ buyers except for her descendants.
Also, the social surplus of the distance-based mechanism
is always as large as those of the two na\"{i}ve ones.
Furthermore, the worst-case inefficiency of the distance-based mechanism can be
bounded by choosing an appropriate reserve price.

\subsection{Bounded Efficiency}
Pareto efficiency in the multi-unit auction with $k$ units
requires that each buyer is a winner only if she is in the set of top-$k$ buyers,
i.e., whose value
is more than or equal to the $k$-th highest value.
However, it is not compatible with strategy-proofness
in our model with the buyers' network, since a buyer would have an
incentive for not forwarding information to her descendants if she needs to
compete with them.
Thus, we introduce a weaker concept called \emph{bounded efficiency},
which is consistent with the incentive of buyers to forward the information.
We say an allocation satisfies bounded efficiency if
each winner is in the set of top-$k$ buyers
\emph{except for its descendants}.
Also, a mechanism satisfies bounded efficiency if it always
obtains a bounded efficient allocation.
By ignoring the descendants of each buyer, the incentive of information forwarding
can still be guaranteed.\footnote{%
Note that the number of buyers, each of which is in the set of top-$k$
buyers except its descendants, can be more than $k$. Thus, it is
impossible to guarantee that all of them are winners.}
Indeed, our mechanism satisfies bounded efficiency.

\begin{proposition}
 \label{prop:top-k}
 The distance-based mechanism satisfies bounded efficiency:
 $\forall \boldsymbol{\theta}'$,
 $\forall i \in \hat{N}$
 s.t.\ $f_{i}(\boldsymbol{\theta}') = 1$,
 $
  \#\{j \in \hat{N}_{-i} \mid v'_{j} > v'_{i}\} < k.
 $
\end{proposition}

\begin{proof}
 Let $i \in \hat{N}$ be an arbitrarily chosen winner
 and $W \subseteq \hat{N} \setminus \{i\}$ be the
 set of winners chosen before $i$ in the mechanism.
 By definition, the winner $i$ faces the price $v^{*}(\hat{N}_{-i} \setminus W, k - |W|)$.
 Since $i$ is a winner, $v'_{i} \geq v^{*}(\hat{N}_{-i} \setminus W, k - |W|)$ holds,
 implying that there are less than $k-|W|$ buyers in $\hat{N}_{-i}
 \setminus W$, whose values are strictly larger than $v'_{i}$, i.e.,
 $
  \#\{j \in \hat{N}_{-i} \setminus W \mid v'_{j} > v'_{i}\} < k - |W|
 $.
 Therefore, regardless of how many winners in $W$ have a strictly larger value
 than $v'_{i}$, it holds that
 $
  \#\{j \in \hat{N}_{-i} \mid v'_{j} > v'_{i}\} < k
 $.
\end{proof}

This property
is useful to show other characteristics
of our mechanism, e.g., Proposition~\ref{prop:EFF:domin}.
One can also easily observe that the two na\"{i}ve mechanisms violate this property.

\subsection{Social Surplus Domination}

A mechanism $(f, p)$ is said to {\em dominate}
another mechanism $(f', p')$
{\em in terms of social surplus}
if for any $N$ and any $\boldsymbol{\theta}'$, it holds that
 $
 \sum_{i \in N} v_{i} \cdot f_{i}(\boldsymbol{\theta}')
 \geq
 \sum_{i \in N} v_{i} \cdot f'_{i}(\boldsymbol{\theta}')
 $.

\begin{proposition}
 \label{prop:EFF:domin}
 The distance-based mechanism dominates both ND-VCG and FCFS-F
 in terms of social surplus, but not vice versa. 
\end{proposition}
\begin{proof}
 When $|\hat{N}| \leq k$, every buyer receives a unit both in the distance-based
 mechanism and in ND-VCG. We then consider the cases of $|\hat{N}| > k$.
 First observe that, when $r_{s} = \hat{N}$,
 i.e., there only exist the direct buyers,
 the set of winners in both mechanisms coincides,
 so that the top-$k$ buyers win a unit; this is obvious from the definition for ND-VCG, and
 it also holds for the distance-based mechanism from
 Proposition~\ref{prop:top-k}.

 Furthermore, 
 when we add new buyers into the network one by one,
 in the ascending order of their distances from the source,
 an original winner becomes a loser in the distance-based mechanism
 only when her value is lower than the value of the new buyer who have just been added.
 Therefore, the arrival of a new buyer then weakly increases the social surplus,
 while it does not change that of ND-VCG.
 Also, there exists a case where
 the social surplus strictly increases.
 Thus, the distance-based mechanism dominates ND-VCG but not vice
 versa.

 In the distance-based mechanism,
 each of the first $k$ buyers, who is a winner in FCFS-F,
 loses only when
 there exists some buyer who arrives later and has a larger value.
 The distance-based mechanism therefore dominates \black{FCFS-F}, but not vice
 versa.
\end{proof}

\subsection{Worst-Case Efficiency Loss}

When the seller wants to maximize revenue,
it is natural to consider introducing a reserve price,
i.e., the threshold bidding value for each buyer
to own the right to win a unit~\cite{myerson:MathOR:1981}.
Letting $v_{h}$ be the reserve price that the seller introduces,
the distance-based mechanism with a reserve price $v_h$
is then implemented by adding $k$ dummy vertices
with value $v_h$ in $T(\boldsymbol{\theta'})$,
each of which is connected only to $s$ (see Fig.~\ref{fig:dummy}),
while in line 2 of the algorithm
the dummies are not considered.
In other words, those dummies only affect $\hat{N}_{-i}$
for each $i \in \hat{N}$ and have no chance to win.
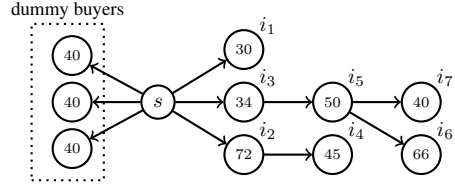
\begin{figure}[t]
\begin{center}
	\begin{tikzpicture}[scale=0.8, transform shape]
	\node[draw, thick, circle] (s) {$s$};
	\node[draw, thick, circle,  label={[xshift = 0.4cm, yshift=-0.2cm] $i_{3}$}] (3) [right = 0.8cm of s] {\scriptsize $34$};
	\node[draw, thick, circle,  label={[xshift = 0.4cm, yshift=-0.2cm] $i_{1}$}] (1) [above = 0.2cm of 3] {\scriptsize $30$};
	\node[draw, thick, circle,  label={[xshift = 0.4cm, yshift=-0.2cm] $i_{2}$}] (2) [below = 0.2cm of 3] {\scriptsize $72$};
	\node[draw, thick, circle] (d) [left = 0.8cm of s] {\scriptsize $40$};
	\node[draw, thick, circle] (dd) [above = 0.1cm of d] {\scriptsize $40$};
	\node[draw, thick, circle] (ddd) [below = 0.1cm of d] {\scriptsize $40$};
	\node[draw, thick, circle,  label={[xshift = 0.4cm, yshift=-0.2cm] $i_{4}$}] (4) [right = 0.8cm of 2] {\scriptsize $45$};
	\node[draw, thick, circle,  label={[xshift = 0.4cm, yshift=-0.2cm] $i_{5}$}] (5) [right = 0.8cm of 3] {\scriptsize $50$};
	\node[draw, thick, circle,  label={[xshift = 0.4cm, yshift=-0.2cm] $i_{6}$}] (6) [right = 0.8cm of 4] {\scriptsize $66$};
	\node[draw, thick, circle,  label={[xshift = 0.4cm, yshift=-0.2cm] $i_{7}$}] (7) [right = 0.8cm of 5] {\scriptsize $40$};

	\draw[draw=black,thick, dotted] (-2.1,-1.3) rectangle (-0.9,1.3);
	\node[above] at (-1.5, 1.25) {\footnotesize dummy buyers};

	\path[draw,thick]
	(s) edge[->]  (1)
	(s) edge[->]  (2)
	(s) edge[->]  (3)
	(s) edge[->]  (d)
	(s) edge[->]  (dd)
	(s) edge[->]  (ddd)
	(3) edge[->]  (5)
	(5) edge[->]  (7)
	(5) edge[->]  (6)
	(2) edge[->]  (4);
\end{tikzpicture}
\end{center}
\caption{Implementation of Reserve Price $v_h (= 40)$
 by Adding $k (= 3)$ Dummy Buyers to Diffusion Critical Tree.}
\label{fig:dummy}
\end{figure}
The following example, which uses the same profile
of the reports with Example~\ref{ex:ex},
demonstrates how the introduction of a reserve price changes the allocation.

\begin{example}
 See Fig.~\ref{fig:dummy}.
 Since there are three units, the mechanism first adds three dummy vertices.
 The price for $i_{1}$ is given as
 $p_{i_1} = 50$,
 and she is not allocated a unit.
 The price for $i_{2}$ is given as
 $p_{i_2} = 45$,
 and she wins a unit.
 The price for $i_{3}$ is given as
 $p_{i_3} = 40$,
 which comes from the valuation of the dummy buyer.
 Since her value is strictly less than $p_{i_{3}}$,
 she is not allocated a unit.
 The price for $i_{4}$ is given as
 $p_{i_4} = 50$,
 and she is not allocated a unit.
 The price for $i_{5}$ is given as
 $p_{i_5} = 40$,
 and she wins a unit. At this moment one unit remains.
 For buyer $i_6$, the price is given as
 $p_{i_6} = 45$,
 which is strictly less than her value of $66$.
 Thus, she wins a unit.
 Now that no unit remains, the price for buyer $i_7$ is set to be infinity,
 To sum up, $i_2$, $i_5$, and $i_6$ win a unit,
 and each pays $40$, $40$, and $45$, respectively.
\end{example}

Nearly identical proofs work for feasibility, non-deficit,
individual rationality, and strategy-proofness.
However, the introduction of a reserve price obviously
breaks down non-wastefulness. Actually, for any non-zero $v_h$,
there is a case where no buyer wins a unit, e.g.,
$v_{i} < v_h$ for every $i \in N$.
This implies that, when we consider the approximation ratio
an efficiency measure, the distance-based mechanism with
reserve price performs poorly.
Even worse, the original definition without a reserve price
still has an arbitrarily worse (i.e., arbitrarily close to zero) approximation ratio.

Nevertheless, it remains important to
clarify the effect of different reserve prices,
given the practical usefulness of reserve prices.
We therefore consider the following worst case efficiency measure
called $\alpha$-inefficiency,
inspired by Nath and Sandholm (\citeyear{nath:GEB:2018}),
and find that the optimal reserve price is $\bar{v}/2$,
where $\bar{v}$ is the upper bound of the value,
i.e., for each $i \in N$, $v_i \leq \bar{v}$.

\begin{definition}
 Let $\bar{v}$ be the upper bound of the value.
 A mechanism $(f,t)$ is {\em $\alpha$-inefficient} if
 \[
 \alpha =
 \frac{1}{k \bar{v}}
 \sup_{\boldsymbol{\theta}' \in \boldsymbol{\Theta}}
 \left[
 \max_{\boldsymbol{x} \in \black{\boldsymbol{X}}} \sum_{i \in N} v'_i \cdot x_i
 - \sum_{i \in N} v'_i \cdot f_i(\boldsymbol{\theta}')
 \right].
\]
\end{definition}

The range of $\alpha$ is $[0,1]$,
and having a smaller $\alpha$ is better.
We first provide a lemma that is useful to provide
the worst-case inefficiency,
while its proof appears in the appendix.
Given $\boldsymbol{\theta'}$,
let $\ell$ denote the number of connected buyers
whose values are
no less than $v_h$,
i.e., $\ell := \# \{i \in \hat{N} \mid v_{i} \geq v_h\}$.

\begin{lemma}
\label{lem:nwf}
 Assuming all buyers declare their true values,
 $\min(\ell, k)$ units are allocated in the distance-based mechanism
 with a reserve price.
\end{lemma}

Given the above lemma, we show that
the distance-based mechanism with $v_h = \bar{v}/2$ satisfies $1/2$-inefficiency.

\begin{theorem}
\label{thm:EFF}
 The distance-based mechanism with reserve price $v_h$ satisfies $1/2$-inefficiency
 by setting $v_h = \bar{v}/2$.
\end{theorem}

\begin{proof}
 If $\ell < k$, the distance-based mechanism allocates units to the top
 $\ell$ buyers within $\hat{N}$ (in terms of values) from Lemma~\ref{lem:nwf}.
 The remaining $k - \ell$ units cannot be allocated
 since the values of other buyers are less than $v_h$.
 Thus, the maximum efficiency loss is bounded by $(k-\ell) v_h$
 (if $k-\ell$ buyers exist whose values are $v_h - \epsilon$,
 the efficiency loss becomes $(k-\ell) (v_h - \epsilon)$).
 In particular, if the value of each buyer is less than
 $v_h$, $\ell$ becomes $0$.
 Thus, the worst case efficiency loss is $k\cdot v_h$.
 If $\ell \geq k$,
 the distance-based mechanism allocates $k$ units from Lemma~\ref{lem:nwf}.
 The maximum efficiency loss is bounded by $k (\bar{v} - v_h)$,
 which can occur, for instance, when there are $2k$ buyers,
 forming a path graph, and those $k$ buyers closer to $s$ have the value
 of $v_h$, while the rest have the value of $\bar{v}$: the
 distance-based mechanism allocates $k$ units to the closest $k$ buyers.
 From the above, the maximum efficiency loss is given as
 $\max(k\cdot v_h, k(\bar{v}-v_h))$.
 This is bounded from the bottom by $k \cdot \bar{v}/2$,
 which is achieved by setting $v_h$ to $\bar{v}/2$.
 Thus, the distance-based mechanism is
 $1/2$-inefficient for $v_h = \bar{v}/2$.
\end{proof}

Observe that there is a tradeoff between achieving non-wastefulness
and guaranteeing a better worst-case performance
by the mechanism with a reserve price,
where the former is achieved by $v_h = 0$
and the latter by $v_h = \bar{v}/2$.
Obtaining the lower bound of $\alpha$
that a strategy-proof mechanism achieves
remains an open question.
However, since $0$-inefficiency implies Pareto efficiency,
the impossibility suggested by Takanashi et al.\
(\citeyear{takanashi:arXiv:2019}) implies that no strategy-proof mechanism that also satisfies non-deficit and
individual rationality achieves $0$-inefficiency.

\section{Revenue Analysis}
\label{sec:revenue}

The seller's revenue is also an important evaluation criterion 
for auction mechanisms.
In this section, we first show that
the seller's revenue in the distance-based mechanism
is no less than those of the two na\"{i}ve ones.
We also show that
maximizing the revenue by optimally choosing the set of
its followers to whom it sends the information is NP-complete.

\subsection{Revenue Domination}

We define the domination in terms of the seller's revenue
analogously.
A mechanism $(f, p)$ {\em dominates} another mechanism $(f', p')$
{\em in terms of the seller's revenue}
if for any $N$ and any $\boldsymbol{\theta}'$, it holds that
$
 \sum_{i \in N} t_{i}(\boldsymbol{\theta}')
 \geq
 \sum_{i \in N} t'_{i}(\boldsymbol{\theta}')
$.

\begin{proposition}
 The distance-based mechanism dominates both ND-VCG and FCFS-F
 in terms of the seller's revenue, but not vice versa.
\end{proposition}

\begin{proof}
 The distance-based mechanism obviously dominates ND-VCG
 when $|r_{s}| \leq k$, since the price for each winner in ND-VCG
 is zero. When $|r_{s}| > k$, each winner in ND-VCG pays $v^{*}(r_{s} \setminus \{i\}, k)$.
 On the other hand, the price $p_{i}$ for each winner $i$ in the distance-based mechanism
 satisfies $p_{i} \geq v^{*}(\hat{N}_{-i}, k)$ by definition.
 For every winner $i$,
 $\hat{N}_{-i}$ is a superset of $r_{s} \setminus \{i\}$.
Therefore, from the monotonicity of $v^{*}$ on the first argument,
 $p_{i} \geq v^{*}(\hat{N}_{-i}, k) \geq v^{*}(r_{s} \setminus \{i\},
 k)$ holds. Also, there exists a case where the inequality becomes strict.
Thus, the distance-based mechanism dominates ND-VCG but not vice versa.
Since the revenue of FCFS-F is always zero, while
the revenue of the distance-based mechanism is non-negative for any
input and can be strictly positive,
 the distance-based mechanism also dominates FCFS-F but not vice versa.
\end{proof}

\subsection{Revenue Monotonicity}

The seller's revenue is required to have some specific form of monotonicity.
Several forms of such {\em revenue monotonicity} have been studied, 
including bidder revenue monotonicity~\cite{rastegari:AIJ:2011,todo:IAT:2010} and
item revenue monotonicity, a.k.a.\ destruction-proofness~\cite{muto:MSS:2017}.

The 
condition studied in this section
is weaker than bidder revenue monotonicity.
A mechanism is {\em follower revenue monotonic}
if the seller's revenue is monotonically increasing
with respect to the number of direct buyers.
Let $t_{i}(\boldsymbol{\theta}' \mid r'_{s})$ be the payment from buyer $i$
when $\boldsymbol{\theta}'$ is reported and
$s$ sends the information of the auction
to a subset $r'_{s}$ of direct buyers.

\begin{definition}
 A mechanism $(f,t)$ is follower revenue monotonic if
 for any $r_{s}$,
 $\boldsymbol{\theta}'$,
 and
 $r'_{s} \subseteq r_{s}$,
 it holds that
 $\sum_{i \in N} t_{i}(\boldsymbol{\theta}' \mid r_{s})
 \geq
 \sum_{i \in N} t_{i}(\boldsymbol{\theta}' \mid r'_{s})$.
\end{definition}

\begin{figure}[t]
 \begin{center}
  \begin{tikzpicture}[scale=0.75, transform shape]
   \node[draw, thick, circle] (s) {$s$};
   \node[draw, thick, circle, label={[xshift = 0.45cm, yshift=-0.3cm] $i_{2}$}] (2) [right = 1cm of s] {\scriptsize $20$};
   \node[draw, thick, circle, inner sep=4.5pt, label={[xshift = 0.45cm, yshift=-0.3cm] $i_{1}$}] (1) [above = 0.1cm of 2] {\scriptsize $5$};
   \node[draw, thick, circle, inner sep=4.5pt, label={[xshift = 0.45cm, yshift=-0.3cm] $i_{3}$}] (3) [below = 0.1cm of 2] {\scriptsize $6$};
   \node[draw, thick, circle, label={[xshift = 0.45cm, yshift=-0.3cm] $i_{4}$}] (4) [right = 1cm of 1] {\scriptsize $15$};

   \path[draw,thick]
   (s) edge[->]  (1)
   (s) edge[->]  (2)
   (s) edge[->]  (3)
   (1) edge[->]  (4);
  \end{tikzpicture}
 \end{center}
\caption{Violation of Revenue Monotonicity.}
\label{fig:rev-mon}
\end{figure}
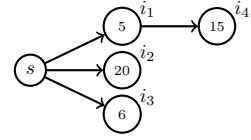
For a multi-unit auction without any network among buyers,
there are several bidder revenue monotonic mechanisms,
and thus follower revenue monotonic ones.
However, the example below shows that,
for auction via social network,
our mechanism is not even follower revenue monotonic.

\begin{example}
Consider two units
and four buyers $i_{1}$, $i_{2}$, $i_{3}$, and $i_{4}$;
see Fig.~\ref{fig:rev-mon}.
The priority order is $i_{1} \succ i_{2} \succ i_{3} \succ i_{4}$.

Assume all of the buyers behave sincerely.
When the seller $s$ sends the information to all of its followers, i.e.,
$i_{1}$, $i_{2}$, and $i_{3}$, all buyers become connected.
The price for each buyer is given as
$p_{1}=6$, $p_{2}=6$, $p_{3}=15$, and $p_{4}=6$,
where $i_2$ and $i_4$ win a unit and the revenue is $12$.
When $s$ sends the information only to $i_{1}$ and $i_{2}$,
$i_{3}$ becomes unconnected.
The price for each buyer is given as
$p_{1}=0$, $p_{2}=15$, and $p_{4}=\infty$,
where $i_{1}$ and $i_{2}$ win a unit
and the revenue is $15$.
As a result, the revenue is not maximized
when the seller sends the information to all the direct buyers.
\end{example}

Therefore, 
a question rises to the seller: to which set of direct buyers
should she send the information to maximize her revenue?
We define a simplified form of this problem,
so called {\sc Optimal Diffusion},
from the perspective of computational complexity:
Assuming that the seller knows the exact network
among buyers and that all buyers behave sincerely,
she tries to find an optimal set of direct buyers
to whom she should send the information.

\begin{definition}[\sc Optimal Diffusion]
 Given a number $k$ of units,
 a profile $\boldsymbol{\theta}'$,
 a set $r_{s}$ of direct buyers,
 and a threshold $K$,
 is there a subset $r'_s \subseteq r_s$ s.t.\
 $\sum_{i \in N} t_{i}(\boldsymbol{\theta}' \mid r'_{s}) \geq K$
 holds under the distance-based mechanism?
\end{definition}

We show that {\sc Optimal Diffusion} is NP-complete
by a reduction from {\sc Partition}.
Due to space limitations, we present a proof sketch below; 
a full proof is in the appendix. 

\begin{definition}[\sc Partition]
 Given a set $\boldsymbol{A}$ where each $i\in \boldsymbol{A}$ has a value $v(i) \in \mathcal{Z}^+$,
 does there exist a subset $\boldsymbol{A'} \subseteq \boldsymbol{A}$ such that
 $\sum_{i\in \boldsymbol{A}'} v(i) = m$, where $m = \sum_{i \in \boldsymbol{A}}
 v(i)/2$?
\end{definition}

\begin{theorem}
 {\sc Optimal Diffusion} is NP-complete.
\end{theorem}

\begin{proof}[\it Proof Sketch]
 First, {\sc Optimal Diffusion} is in NP
 since we can compute $\sum_{i \in N} t_i(\boldsymbol{\theta}'\! \mid \! r'_s)$ in polynomial time.
 Given an instance of {\sc Partition}, we construct an instance of
 {\sc Optimal Diffusion} as follows,
 with $N = N^A \cup N^B \cup N^C$:
 \begin{itemize}
  \item For all $i \!\! \in \!\! \boldsymbol{A}$, we create set $(a^i_j)_{0 \leq j \leq v(i)}$ in $N^A$ such that
	  ${\theta_{a^i_0} \! = \! (\epsilon, (a^i_j)_{1 \leq j \leq v(i)})}$, and
        $\theta_{a^i_j} \! = \! (v_1, \emptyset)$ for $1 \! \leq \! j \! \leq \! v(i)$.
  \item Set $N^B = (b_j)_{1 \leq j \leq m+2}$ is such that
        $\theta_{b_1}= (v_2, \{b_2\})$,
        $\theta_{b_j}= (v_3, \{b_{j+1}\})$ for $2 \! \leq \! j \! \leq \! m\!+\!1$, and
        $\theta_{b_{m+2}}= (v_4, \emptyset)$.
  \item Set $N^C = (c_j)_{1 \leq j \leq m+1}$ is such that
	  $\theta_{c_j}= (\epsilon, \{c_{j+1}\})$ for $1\! \leq \! j \!\leq \!m$, and
        $\theta_{c_{m+1}}= (v_5, \emptyset)$.
\item The seller's direct followers are $(a^i_0)_{i \in A} \cup \{b_1, c_1\}$.
 \end{itemize}
The network is illustrated in Fig. \ref{fig:SketchOptAuction}.
\begin{figure}[t]
\begin{center}
	\begin{tikzpicture}[scale=.9, transform shape]

	\node[draw, thick, circle] (s) {$s$};
	\node[draw, thick, circle,  label=above:{$a^z_0$}] (a2) [right = 0.8cm of s] {\small $\epsilon$};
	\node[draw, thick,  circle,  label=above:{$a^1_0$}] (a1) [above = 1.5cm of a2] {\small $\epsilon$};
	\node[draw, thick, circle,  label={[xshift = 0.4cm, yshift=-0.2cm] $b_1$}] (b1) [below = 1.1cm of a2] {\small $v_2$};
	\node[draw, thick, circle,  label={[xshift = 0.4cm, yshift=-0.2cm] $c_1$}] (c1) [below = 0.3cm of b1] {\small $\epsilon$};
	\node[draw, thick,  circle,  label={[xshift = 0.55cm, yshift=-0.4cm] $a^1_1$}] (a11) [above right = 0.15cm and 1cm of a1] {\small $v_1$};
	\node[draw, thick,  circle,  label={[xshift = 0.8cm, yshift=-0.5cm] $a^1_{v(a^1)}$}] (a12) [below right = 0.15cm and 1cm of a1] {\small $v_1$};
	\node[draw, thick,  circle,  label={[xshift = 0.55cm, yshift=-0.4cm] $a^z_1$}] (a21) [above right = 0.15cm and 1cm of a2] {\small $v_1$};
	\node[draw, thick,  circle,  label={[xshift = 0.8cm, yshift=-0.5cm] $a^z_{v(a^z)}$}] (a22) [below right = 0.15cm and 1cm of a2] {\small $v_1$};
	\node[draw, thick,  circle,  label={[xshift = 0.4cm, yshift=-0.2cm] $b_2$}] (b2) [right = 0.4cm of b1] {\small $v_3$};
	\node[draw, thick,  circle,  label={[xshift = 0.4cm, yshift=-0.2cm] $c_2$}] (c2) [right = 0.4cm of c1] {\small $\epsilon$};
	\node[draw, thick,  circle,  label={[xshift = 0.6cm, yshift=-0.2cm] $b_{m+1}$}] (bm) [right = 1.4cm of b2] {\small $v_3$};
	\node[]  (bminter) [right = 0.4cm of b2] {\footnotesize $\ldots$};
	\node[]  (cminter) [right = 0.25cm of c2] {\footnotesize $\ldots$};
	\node[draw, thick,  circle,  label={[xshift = 0.4cm, yshift=-0.2cm] $c_{m}$}] (cm1) [right = 1.2cm of c2] {\small $\epsilon$};
	\node[draw, thick,  circle,  label={[xshift = 0.6cm, yshift=-0.2cm] $b_{m+2}$}] (bm2) [right = 0.6cm of bm] {\small $v_4$};
	\node[draw, thick,  circle,  label={[xshift = 0.6cm, yshift=-0.2cm] $c_{m+1}$}] (cm2) [right = 0.4cm of cm1] {\small $v_5$};

	\path[draw,thick]
	(s) edge[->]  (a1)
	(s) edge[->]  (a2)
	(s) edge[->]  (b1)
	(s) edge[->]  (c1)
	(a1) edge[->]  (a11)
	(a1) edge[->]  (a12)
	(a2) edge[->]  (a21)
	(a2) edge[->]  (a22)
	(b1) edge[->]  (b2)
	(c1) edge[->]  (c2)
	(cm1) edge[->]  (cm2)
	(b2) edge[->]  (bminter)
	(bminter) edge[->]  (bm)
	(c2) edge[->]  (cminter)
	(cminter) edge[->]  (cm1)
	(bm) edge[->]  (bm2);

	\path
	(a1)-- node[auto=false, pos=0.3]{\footnotesize \vdots}  (a2)
	(a11)-- node[auto=false, pos=0.3]{\footnotesize \vdots}  (a12)
	(a21)-- node[auto=false, pos=0.3]{\footnotesize \vdots}  (a22);
\draw [decorate, decoration={brace,amplitude=6pt,raise=0pt},yshift=0pt]
(4.5,2.8) -- (4.5,1.2) node [black,midway,xshift=1.2cm,yshift=0cm] {\footnotesize $v(a^1)$ buyers};
\draw [decorate, decoration={brace,amplitude=6pt,raise=0pt},yshift=0pt]
(7.2,2.8) -- (7.2,-0.9) node [black,midway,xshift=0.6cm,yshift=0cm] {\footnotesize $N^A$};
\draw [decorate, decoration={brace,amplitude=6pt,raise=0pt},yshift=0pt]
(4.5,0.8) -- (4.5,-0.8) node [black,midway,xshift=1.2cm,yshift=0cm] {\footnotesize $v(a^z)$ buyers};
\draw [decorate, decoration={brace,amplitude=6pt,raise=0pt},yshift=0pt]
(7.2,-1.3) -- (7.2,-2.0) node [black,midway,xshift=0.6cm,yshift=0cm] {\footnotesize $N^B$};
\draw [decorate, decoration={brace,amplitude=6pt,raise=0pt},yshift=0pt]
(7.2,-2.3) -- (7.2,-3.0) node [black,midway,xshift=0.6cm,yshift=0cm] {\footnotesize $N^C$};
\end{tikzpicture}
\end{center}
\caption{Reduction from Partition: Network of Buyers.}
\label{fig:SketchOptAuction}
\end{figure}
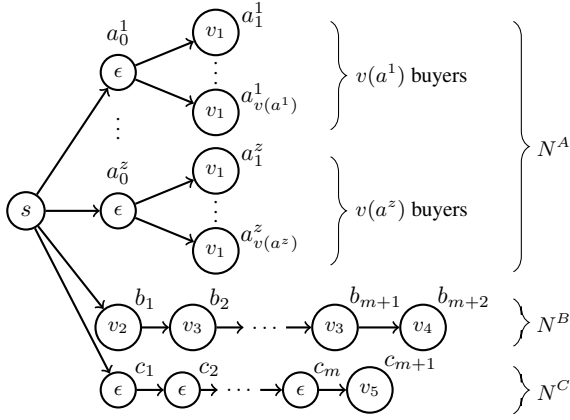
Buyers are labelled with any ascending order of $d(\cdot)$ satisfying
$b_{m+1} \! \succ \! c_{m+1}$.
The prices satisfy
$\epsilon \! <\! v_2\! <\! v_1\! <\! v_3\! <\!<\! v_4\! <\! v_5$.
The number of units is $k\!=\! m+2$ and the threshold is $K\! =\! \epsilon + m\cdot v_1 + v_4$.

We briefly argue the validity of the reduction.
Notice that buyers $b_1$ and $c_1$ belong to any $r'_s \! \subseteq \! r_s $
such that $\sum_{i \in N} t_i(\black{\boldsymbol{\theta}} \mid r'_s) \geq K$,
since otherwise price $v_4$ cannot be reached.

If $|\{a \! \in \! r'_s \mid v_a \!= \!v_1\}| \! = \!m$ holds, i.e.,
exactly $m$ descendants with
value $v_1$ can be chosen (thus the original {\sc Partition} is ``yes''),
then buyer $b_1$ buys at price $\epsilon$,
buyers $(b_i)_{2 \leq i \leq m+1}$ at price $v_1$, and buyer $c_{m+1}$
at price $v_4$.
Hence, $\sum_{i \in N} t_i(\black{\boldsymbol{\theta}} \mid r'_s) = \epsilon +
m\cdot v_1 + v_4=K$ and {\sc Optimal Diffusion} is ``yes''.

If the original {\sc PARTITION} is ``no'', either
(i) $| {a \in r'_s \mid v_a = v_1}|< m$ or
(ii) $| {a \in r'_s \mid v_a = v_1}|> m$  holds.
In the case (i), buyers $b_1$ and $b_2$ buy at price $\epsilon$.
Hence, $\sum_{i \in N} t_i(\black{\boldsymbol{\theta}} \mid r'_s) < K$ and
 {\sc Optimal Diffusion} is ``no''.

 In the case (ii), buyer $b_1$ does not buy, and $c_{m+1}$ buys at
 price lower than $v_4$.
Hence, $\sum_{i \in N} t_i(\black{\boldsymbol{\theta}} \mid r'_s) < K$ and
 {\sc Optimal Diffusion} is ``no''.
\end{proof}

\section{Incentive Analysis}
\label{sec:incentive}

Now we show that, compared with those two na\"{i}ve mechanisms,
our mechanism also has its own strength on buyers incentive;
in those mechanisms, hiding the information, combined with the report of the true value,
is also a dominant strategy,
while this is not the case in our mechanism for any $k \geq 2$.
This indicates that the incentive for each buyer
to report her type truthfully in the distance-based mechanism is stronger than
that in both of those na\"{i}ve ones.

\begin{proposition}
 Assume $k \geq 2$.
 For each $i$, reporting $(v_{i}, \emptyset)$
 is not a dominant strategy in the distance-based mechanism.
\end{proposition}

\begin{proof}
Consider $k$ units and $k+2$ buyers
 $i_{1}, \ldots, i_{k+2}$, such that
 $r_{s} = \{i_{1}, i_{3}, i_{5}, i_{6}, \ldots, i_{k+2}\}$,
 $\theta_{i_{1}} = (15, \{i_{2}\})$,
 $\theta_{i_{2}} = (20, \emptyset)$,
 $\theta_{i_{3}} = (10, \{i_{4}\})$,
$\theta_{i_{4}} = (9, \emptyset)$,
and $\theta_{i_{j}} = (30, \emptyset)$ for all $5 \leq j \leq k+2$.
The priority is given as $i_{5} \succ i_{6} \succ \cdots \succ i_{k+2}
\succ i_{3} \succ i_{1} \succ i_{4} \succ i_{2}$.
The first $k-2$ units are sold to $\{i_{j}\}_{5 \leq j \leq k+2}$,
regardless of $i_{1}$'s forwarding strategy.
 Under $i_{1}$'s sincere forwarding to $i_{2}$, $i_{1}$ wins a unit
and pays $9$.
 If $i_{1}$ does not forward the information to $i_{2}$,
 then $i_{1}$ would win a unit and pay $10$.
So not forwarding the information is dominated by a sincere forwarding
 in this case.
\end{proof}

\begin{proposition}
 For each $i$,
 reporting $(v_{i}, \emptyset)$
 is a dominant strategy in both ND-VCG and FCFS-F.
\end{proposition}

\begin{proof}
 In ND-VCG, only the reports from the direct buyers
 affect the outcome. Therefore,
 for each direct buyer $i \in r_{s}$,
 any valuation report $v'_{i}$,
 and any action by other buyers,
 the choice of $r'_{i} \subseteq r_{i}$ does not change the outcome at all.

 In FCFS-F,
 for each $i \in N$,
 any follower of $i$ who originally arrives after $i$ under $i$'s sincere forwarding $r_{i}$
 is still arriving after $i$ under any manipulation $r'_{i} \subset r_{i}$.
 Thus, whether $i$ wins a unit does not depend on the choice of $r'_{i}$.
\end{proof}

\section{Conclusions}
\label{conclu}

The distance-based mechanism satisfies strategy-proofness,
non-wastefulness, non-deficit, and individual rationality.
The performance is comprehensively analyzed; it dominates
the two na\"{i}ve mechanisms in terms of both social surplus and revenue.
Several other properties are also revealed.

A more detailed analysis on the complexity of maximizing the
seller's revenue is required, such as for the case with a fixed number $k$
of units.
Our future work will also include
more general revenue analysis,
e.g., revenue equivalence~\cite{heydenreich:ECTA:2009} and
revenue optimality~\cite{myerson:MathOR:1981}.
Extending the distance-based mechanism for more general domains,
such as multi-unit auctions with decreasing marginal values,
is also crucial.
Considering an obviously strategy-proof auction via social network
will also be an interesting direction~\cite{li:AER:2017}.

\section{Acknowledgments}
\label{ack}
This work is partially supported by JSPS KAKENHI Grants JP17H00761 and JP17H04695, and JST SICORP JPMJSC1607.

\clearpage

\appendix 
\section{Discussion on Revelation Principle}
\label{app:revelation}
The {\em revelation principle} is a fundamental concept
of mechanism design, which specifies the relation between direct revelation
and indirect mechanisms.
If an indirect mechanism has a dominant strategy equilibrium,
then we can find a strategy-proof direct revelation mechanism
that results in the same outcome in the equilibrium.
%
However, in a non-standard setting,
e.g., online mechanism design~\cite{parkes:AGT:2007},
the revelation principle might fail.
Thus, we need to exert caution.

First, a direct revelation mechanism should not utilize any
information that is unavailable in an indirect mechanism.
On one hand, in our setting, 
the seller does not know the existence of the buyers
who are not directly connected to her.
On the other hand, in a direct revelation mechanism,
we assume the mechanism designer knows $N$,
i.e., all the buyers, and gets reports from each buyer in $N$.
We require that in a feasible allocation, a direct revelation mechanism
can only allocate a unit to a connected buyer.
Thus, a feasible direct revelation mechanism
does not utilizes any information
that is unavailable in an indirect mechanism.

Second, the space of possible manipulations in a direct revelation mechanism
must be {\em equivalent} to that in an indirect mechanism.
In our setting, a buyer can strategically stop forwarding
the information to her followers. By assuming that
buyer $i$, whose true followers are $r_i$, can report only
$r'_i \subseteq r_i$, we preserve the equivalence.
However, just preserving the equivalence is inadequate.
An intuitive explanation why the revelation principle holds
is as follows. Assume an indirect mechanism with a
dominant strategy equilibrium. Then we can create a direct revelation mechanism,
by introducing a mediator between the indirect mechanism and a
participant. The mediator asks for the participant's true type and
plays her dominant strategy. So she does not have an incentive
for lying to the mediator.
When the possible manipulations/reportable types
are restricted, this is no longer true.
Intuitively, assume a participant misreports
her type to the mediator.
Then the mediator does the best manipulation based on the {\em reported} type.
As a result,
the space of possible manipulations in the direct revelation mechanism
can be expanded more compared to that of the indirect mechanism.


Green and Laffont (\citeyear{green:RES:1986}) proposed a condition on reportable types
called the {\em nested range condition} (NRC) and showed that
it is actually a necessary and sufficient condition for the revelation principle
to hold for the case of mechanism design with transfer functions, such as auctions.
Intuitively, NRC is a transitivity condition, which requires that
if an agent (buyer) with type $\theta_{i}$ can pretend to have type $\theta_{i}'$
and
if an agent with type $\theta_{i}'$ can pretend to have type $\theta_{i}''$,
then the agent with type $\theta_{i}$ can also pretend to have type $\theta_{i}''$.
In our model, NRC obviously holds, since the restriction on
$r_{i}$ satisfies transitivity, and there is no restriction on
reportable valuation functions. Therefore, to discuss the outcome
in a dominant strategy equilibrium in an auction via social network,
focusing on direct revelation mechanisms is without loss of generality.

\section{The Formal Definitions of Two Na\"{i}ve Mechanisms}
\label{app:def-mechanisms}

\begin{definition}[ND-VCG]
 Given $\boldsymbol{\theta}'$,
 No-Diffusion-VCG $(f,t)$ is defined as follows:
 \[
 f_{i}(\boldsymbol{\theta}') =
 \begin{cases}
  1 & \text{ if } i \in r_{s} \land v'_{i} \geq v^{*}(r_{s} \setminus \{i\}, k)\\
  0 & \text{ otherwise.}
 \end{cases}
 \]
 \[
 t_{i}(\boldsymbol{\theta}') =
 \begin{cases}
  v^{*}(r_{s} \setminus \{i\}, k) & \text{ if } f_{i}(\boldsymbol{\theta}') = 1 \\
  0 & \text{ otherwise.}
 \end{cases}
 \]
\end{definition}

\begin{definition}[FCFS-F]
 Given $\boldsymbol{\theta}'$,
 first label the connected buyers $\hat{N}$ in
 an ascending order of $d(\cdot)$,
 as we did in the proposed mechanism (described in Definition 6).
 Assume w.l.o.g.\ that the order is $1, 2, \ldots, |\hat{N}|$.
 Then, FCFS-for-Free $(f,t)$ is defined as follows:
 \[
  f_{i}(\boldsymbol{\theta}') =
 \begin{cases}
  1 & \text{ if } i \leq k\\
  0 & \text{ otherwise.}
 \end{cases}
 \]
 \[
  t_{i}(\boldsymbol{\theta}') = 0 \text{ for all } i \in N
 \]
\end{definition}

\section{Omitted Proofs for Section 3}
\label{app:section3}

\begin{proof}[Proof of Theorem 1]
 Let $(f,t)$ be the proposed mechanism.

 {\bf Feasibility:}
 Since $v^{*}(\cdot, k') = \infty$
 for any $k' \leq 0$,
 and $k'$ is initialized by $k$ and decremented whenever a unit is sold
 (see lines 4--6 of the procedure),
 at most $k$ connected buyers buy a unit,
 i.e., $\sum_{i \in N} f_{i}(\boldsymbol{\theta}') \leq k$ for any $\boldsymbol{\theta}'$.
 Furthermore, since any unconnected buyer has no chance to buy a unit,
 feasibility is guaranteed.

 {\bf Individual Rationality:}
 Each buyer $i$ faces a price $p_{i}$
 and buys a unit only when $v'_{i} \geq p_{i}$.
 Thus, under truth-telling, each buyer's price
 is not strictly larger than the true value $v_{i}$
 if she wins a unit, i.e., $v_{i} \geq p_{i}$,
 which guarantees individual rationality.

 {\bf Non-Deficit:}
 Furthermore, by the definition of $v^{*}$,
 the price for each winner is non-negative.
 Since each loser receives no compensation,
 it also satisfies non-deficit.
\end{proof}

\begin{proof}[Proof of Theorem 2]
 Let $(f,t)$ be the distance-based mechanism.
 Consider an arbitrary report profile $\boldsymbol{\theta}'$. 
 If $|\hat{N}| \leq k$, the price that each buyer $i \in \hat{N}$ faces is
 $v^*(\hat{N}_{-i} \setminus \hat{W}_{\succ i}, k - |\hat{W}_{\succ i}|)$.
 Since
 $|\hat{N}_{-i} \setminus \hat{W}_{\succ i}| <
 |\hat{N}| - |\hat{W}_{\succ i}| \leq
 k - |\hat{W}_{\succ i}|$, the price is $0$.
 Thus, 
 the number of winners is $|\hat{N}| = \min(k, |\hat{N}|)$,
 satisfying non-wastefulness.
 We then focus on the case of $k < |\hat{N}|$.
 Assume for the sake of contradiction that
 the number of units sold under $\boldsymbol{\theta}'$ is
 strictly smaller than $k$.
 More specifically,
 letting $m := \sum_{i \in \hat{N}} f_{i}(\boldsymbol{\theta}')$,
 we assume that $m < k$ holds.

 Assume w.l.o.g.\ that the top $m+1$ buyers among $\hat{N}$,
 in terms of reported values,
 are labeled as $i_{1}, i_{2}, \ldots, i_{m+1}$
 in a descending order of reported value $v'$,
 with an arbitrary tie-breaking, i.e.,
 $v'_{i_{1}} \geq v'_{i_{2}} \geq \ldots \geq v'_{i_{m+1}}$.
 Also let $M$ be the top $m$ buyers,
 i.e., $M := \{i_{1}, \ldots, i_{m}\}$.
 First, we show by induction that each member of $M$ must be a winner.
 For the base case, let us show that $i_1$ obtains a unit.
 The price of $i_1$ is given as
 $v^*(N_{-i_1} \setminus \hat{W}_{\succ i_1}, k - |\hat{W}_{\succ i_1}|)$.
 Since $k - |\hat{W}_{\succ i_1}| \geq 1$ (note that we assume the number of
 winners is fewer than $k$), this value  is less than or equal to
 $v^*(N_{-i_1}, 1) \leq v'_{i_1}$.
 Thus, $i_1$ is a winner.
 For the inductive case, let us assume
 $M'=\{i_1, \ldots, i_g\}$ are winners.
 We show that $i_{g+1}$ is also a winner.
 Let $M''$ denote $\{i \in M' \mid d(i) < d(i_{g+1})\}$.
 Also, let $g_1$ denote $|M''|$ and
 $g_2$ denote $g - g_1$.
 Since $i_{g+1}$ is the $g+1$-st buyer and $M'$ is the set of top $g$
 buyers, 
 $
 v'_{i_{g+1}} = v^*(\hat{N}\setminus W', 1+g - w')
 $
 holds
 for any subset $W' \subseteq M'$, where $w' := |W'|$.
 Furthermore, since the function $v^*$ is monotonically non-decreasing
 on the first argument, $\hat{N} \supseteq \hat{N}_{-i_{g+1}}$ implies
 $
 v^*(\hat{N}\setminus W', 1+g - w')\geq v^*(\hat{N}_{-i_{g+1}}\setminus W', 1+g - w')
 $.
 The price of $i_{g+1}$ is given as
 $v^*(\hat{N}_{-i_{g+1}} \setminus \hat{W}_{\succ i_{g+1}}, k - |\hat{W}_{\succ i_{g+1}}|)$.
 Since $k - |\hat{W}_{\succ i_{g+1}}| \geq g_2 +1$ and
 $\hat{W}_{\succ i_{g+1}}$ is a superset of $M''$,
 the price is smaller than or equal to
 $v^*(\hat{N}_{-i_{g+1}} \setminus M'', g_2 +1)$.
 Combining with the above inequality (where we choose $W'$ to $M''$,
 s.t.\ $|M''|=g_1$ and $g_2 = g - g_1$),
 the price is smaller than or equal to
 $v^*(\hat{N} \setminus M'', 1 + g - |M''|)$, which equals
 $v'_{i_{g+1}}$.
 Thus, $i_{g+1}$ must be a winner.
 Therefore, each member of $M$ is allocated a unit.
 Also, since we assume the number of winners is $m$, only the members of $M$ obtain a unit.

 Now consider the buyer $i_{m+1}$.
 Since $M$ is the set of top $m$ buyers
 and $i_{m+1}$ is the $m+1$-st buyer,
 $
 v'_{i_{m+1}} = 
 v^{*}(\hat{N} \setminus W', 1 + m - w')
 $ holds
 for any subset $W' \subseteq M$, where $w' := |W'|$.
 Furthermore, since the function $v^{*}$ is monotonically non-decreasing
 on the first argument, $\hat{N} \supseteq \hat{N}_{-i_{m+1}}$ implies
 $
 v^{*}(\hat{N} \setminus W', 1 + m - w') \geq v^{*}(\hat{N}_{-i_{m+1}} \setminus W', 1 + m - w')
 $.

 The price that the buyer faces is
 given as
 $v^{*}(\hat{N}_{-i_{m+1}} \setminus \hat{W}_{\succ i_{m+1}}, k - |\hat{W}_{\succ i_{m+1}}|)$.
 Since $m < k$ and both $m$ and $k$ are integers,
 $m+1 \leq k$ holds.
 Thus, we have
 $v^{*}(\hat{N}_{-i_{m+1}} \setminus \hat{W}_{\succ i_{m+1}}, 1 + m - |\hat{W}_{\succ i_{m+1}}|)
 \geq
 v^{*}(\hat{N}_{-i_{m+1}} \setminus \hat{W}_{\succ i_{m+1}}, k - |\hat{W}_{\succ i_{m+1}}|)$
 from the definition of $v^{*}$.
 Since $\hat{W}_{\succ i_{m+1}}$ is also a subset of $M$,
 from the above inequalities, we finally have
 $
 v'_{i_{m+1}}
 \geq
 v^{*}(\hat{N}_{-i_{m+1}} \setminus \hat{W}_{\succ i_{m+1}}, k - |\hat{W}_{\succ i_{m+1}}|)
 $,
 where RHS is the price that $i_{m+1}$ faces.
 Therefore, she also wins a unit, violating
 the assumption that only $m$ units are sold.
\end{proof}

\section{Omitted Proof for Section 4}
\label{app:section4}

\begin{proof}[Proof of Lemma 3]
 If $\ell \geq k$,
 $k$ units are allocated 
 from Theorem 2.
 Assume $\ell < k$.
 Let $M = \{i_1, \ldots, i_{\ell}\}$ denote the set of buyers whose
 evaluation values are more than or equal to $v_r$.
 Without loss of generality, we assume $v_{i_1} \geq v_{i_2} \geq
 \ldots \geq v_{i_{\ell}}$ holds,
 and show that each buyer in $M$ is allocated a unit
 by induction.
 For the base case, we show that $i_1$ obtains a unit.
 The price of $i_1$ is given as
 $v^*(N_{-i_1} \setminus W_{\succ i_1}, k - |W_{\succ i_1}|)$.
 Since $k - |W_{\succ i_1}| \geq 1$, this value
 is less than or equal to
 $v^*(N_{-i_1}, 1) \leq v_{i_1}$.
 Thus, $i_1$ is a winner.

 For the inductive case, assume
 $M'=\{i_1, \ldots, i_g\}$ are winners
 for any $g < \ell$.
 We show that $i_{g+1}$ is also a winner.
 Let $M''$ denote $\{i \in M' \mid d(i) < d(i_{g+1})\}$.
 Also, let $g_1$ denote $|M''|$ and
 $g_2$ denote $g - g_1$.
 Since $i_{g+1}$ is the $g+1$-st buyer and $M'$ is the set of top $g$
 buyers, it holds that
 \[v'_{i_{g+1}} = v^*(\hat{N} \setminus W', 1+g - w'),\]
 for any subset $W' \subseteq M'$, where $w' := |W'|$.
 Furthermore, since the function $v^*$ is monotonically non-decreasing
 on the first argument, $\hat{N} \supseteq \hat{N}_{-i_{g+1}}$ implies
 \[v^*(\hat{N}\setminus W', 1+g - w')\geq v^*(\hat{N}_{-i_{g+1}} \setminus W', 1+g - w').\]
 The price of $i_{g+1}$ is given as
 $v^*(\hat{N}_{-i_{g+1}} \setminus \hat{W}_{\succ i_{g+1}}, k - |\hat{W}_{\succ i_{g+1}}|)$.
 Since $k - |\hat{W}_{\succ i_{g+1}}| \geq g_2 +1$ and
 $\hat{W}_{\succ i_{g+1}}$ is a superset of $M''$,
 the price is smaller than or equal to
 $v^*(\hat{N}_{-i_{g+1}} \setminus M'', g_2 +1)$.
 Combining with the above inequality (where we choose $W'$ to $M''$,
 s.t.\ $|M''|=g_1$ and $g_2 = g - g_1$),
 we can see that the price is smaller than or equal to
 $v^*(\hat{N}\setminus M'', 1 + g - |M''|)$, which equals
 $v'_{i_{g+1}}$.
 Thus, $i_{g+1}$ must be a winner.
 Therefore, all the top $\ell$ buyers are allocated a unit.
\end{proof}

\section{Omitted Proof for Section 5}
\label{app:section5}

\begin{proof}[Proof of Theorem 5]
 First, {\sc Optimal Diffusion} is in NP since it is polynomial
 to compute $\sum_{i \in N} t_i(\boldsymbol{\theta'} \mid r'_s)$.
 We show that this problem is NP-complete by a reduction from {\sc Partition}.
 Given an instance of {\sc Partition}, we construct an instance of {\sc Optimal Diffusion} as follows.
 First, let us partition the set of buyers in three subsets $N = N^A \cup N^B \cup N^C$.
 \begin{itemize}
 	\item For each $i \in \boldsymbol{A}$, we create $(v(i)+1)$ buyers
 	$(a^i_j)_{j\in [0,\ldots, v(i)]}$ in $N^A$ such that:
 		\begin{itemize}
 			\item $\theta'_{a^i_0}=(\epsilon, (a^i_j)_{j\in [1, \ldots,v(i)]})$,
 			\item $\theta'_{a^i_j}=(v_1, \emptyset)$ for all $j \in [1, \ldots, v(i)]$.
 		\end{itemize}
 	\item Set $N^B = (b_j)_{j\in [1, \ldots, m+2]}$ is composed of $(m+2)$ buyers such that:
 		\begin{itemize}
 			\item $\theta'_{b_1}= (v_2, \{b_2\})$,
 			\item $\theta'_{b_j}= (v_3, \{b_{j+1}\})$ for all $j\in [2, \ldots, m+1]$,
 			\item $\theta'_{b_{m+2}}= (v_4, \emptyset)$.
 		\end{itemize}
 	\item Set $N^C = (c_j)_{j\in [1, \ldots, m+1]}$ is composed of $(m+1)$ buyers such that:
 		\begin{itemize}
 			\item $\theta'_{c_j}= (\epsilon, \{c_{j+1}\})$ for all $j\in [1, \ldots, m]$,
 			\item $\theta'_{c_{m+1}}= (v_5, \emptyset)$.
 		\end{itemize}
 \end{itemize}
 The network construction is illustrated in Figure \ref{fig:OptAuction}.

 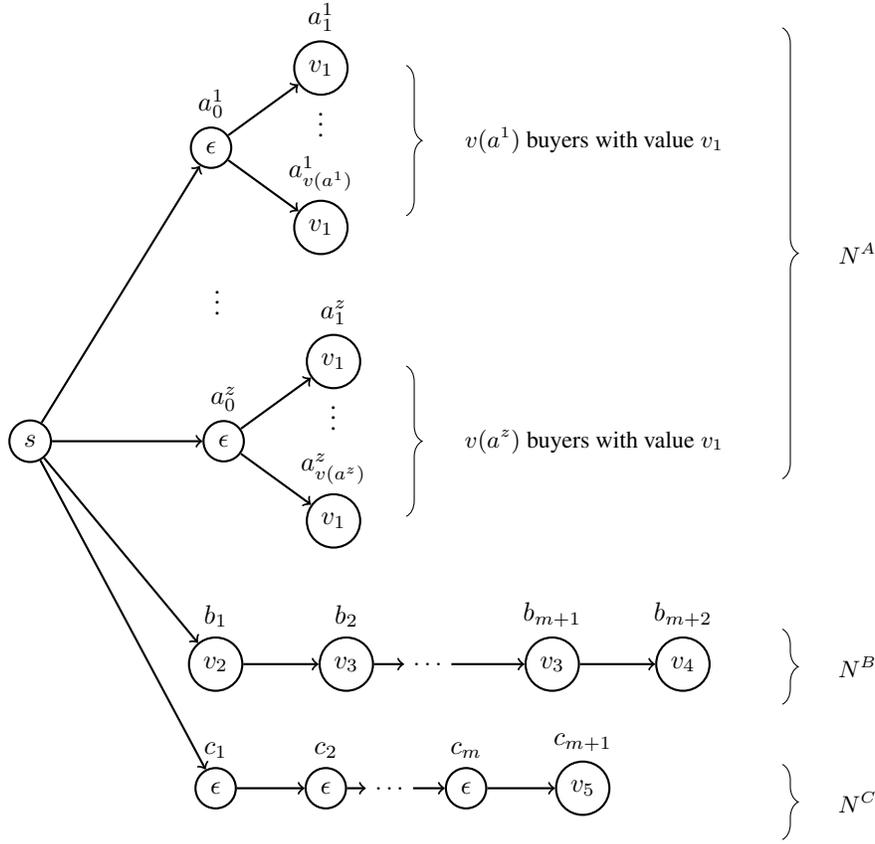
\begin{figure*}[t]
 \begin{center}
 \begin{tikzpicture}
 	\node[draw, thick, circle] (s) {$s$};
 	\node[draw, thick,  circle,  label=above:{$a^1_0$}] (a1) [above right = 3.5cm and 2cm of s] {$\epsilon$};
 	\node[draw, thick, circle,  label=above:{$a^z_0$}] (a2) [right = 2cm of s] {$\epsilon$};
 	\node[draw, thick, circle,  label=above:{$b_1$}] (b1) [below right = 2.5cm and 2cm of s] {$v_2$};
 	\node[draw, thick, circle,  label=above:{$c_1$}] (c1) [below =1cm of b1] {$\epsilon$};
 	\node[draw, thick,  circle,  label=above:{$a^1_1$}] (a11) [above right = 0.6cm and 1cm of a1] {$v_1$};
 	\node[draw, thick,  circle,  label=above:{$a^1_{v(a^1)}$}] (a12) [below right = 0.6cm and 1cm of a1] {$v_1$};
 	\node[draw, thick,  circle,  label=above:{$a^z_1$}] (a21) [above right = 0.6cm and 1cm of a2] {$v_1$};
 	\node[draw, thick,  circle,  label=above:{$a^z_{v(a^z)}$}] (a22) [below right = 0.6cm and 1cm of a2] {$v_1$};
 	\node[draw, thick,  circle,  label=above:{$b_2$}] (b2) [right = 1cm of b1] {$v_3$};
 	\node[draw, thick,  circle,  label=above:{$c_2$}] (c2) [right = 0.9cm of c1] {$\epsilon$};
 	\node[draw, thick,  circle,  label=above:{$b_{m+1}$}] (bm) [right = 2cm of b2] {$v_3$};
	\node[]  (bminter) [right = 0.4cm of b2] {\footnotesize $\ldots$};
	\node[]  (cminter) [right = 0.25cm of c2] {\footnotesize $\ldots$};
 	\node[draw, thick,  circle,  label=above:{$c_{m}$}] (cm1) [right = 1.3cm of c2] {$\epsilon$};
 	\node[draw, thick,  circle,  label=above:{$b_{m+2}$}] (bm2) [right = 1cm of bm] {$v_4$};
 	\node[draw, thick,  circle,  label=above:{$c_{m+1}$}] (cm2) [right = 0.9cm of cm1] {$v_5$};

 	\path[draw,thick]
 	(s) edge[->]  (a1)
 	(s) edge[->]  (a2)
 	(s) edge[->]  (b1)
 	(s) edge[->]  (c1)
 	(a1) edge[->]  (a11)
 	(a1) edge[->]  (a12)
 	(a2) edge[->]  (a21)
 	(a2) edge[->]  (a22)
 	(b1) edge[->]  (b2)
 	(c1) edge[->]  (c2)
	(b2) edge[->]  (bminter)
	(bminter) edge[->]  (bm)
	(c2) edge[->]  (cminter)
	(cminter) edge[->]  (cm1)
 	(cm1) edge[->]  (cm2)
 	(bm) edge[->]  (bm2);

 	\path
 	(a1)-- node[auto=false]{\vdots}  (a2)
 	(a11)-- node[auto=false, pos=0.2]{\vdots}  (a12)
 	(a21)-- node[auto=false, pos=0.2]{\vdots}  (a22);
 \draw [decorate, decoration={brace,amplitude=6pt,raise=0pt},yshift=0pt]
 (5,5) -- (5,3) node [black,midway,xshift=2.5cm,yshift=0cm] {\footnotesize $v(a^1)$ buyers with value $v_1$};
 \draw [decorate, decoration={brace,amplitude=6pt,raise=0pt},yshift=0pt]
 (10,5.5) -- (10,-0.5) node [black,midway,xshift=1cm,yshift=0cm] {\footnotesize $N^A$};
 \draw [decorate, decoration={brace,amplitude=6pt,raise=0pt},yshift=0pt]
 (5,1) -- (5,-1) node [black,midway,xshift=2.5cm,yshift=0cm] {\footnotesize $v(a^z)$ buyers with value $v_1$};
 \draw [decorate, decoration={brace,amplitude=6pt,raise=0pt},yshift=0pt]
 (10,-2.5) -- (10,-3.5) node [black,midway,xshift=1cm,yshift=0cm] {\footnotesize $N^B$};
 \draw [decorate, decoration={brace,amplitude=6pt,raise=0pt},yshift=0pt]
 (10,-4.3) -- (10,-5.3) node [black,midway,xshift=1cm,yshift=0cm] {\footnotesize $N^C$};
 \end{tikzpicture}
 \end{center}
 \caption{Reduction from {\sc Partition}: Network and Values of Buyers.}
 \label{fig:OptAuction}
 \end{figure*}
 The prices satisfy the following condition:
 \[
 	\epsilon < v_2 < v_1 < v_3 << v_4 < v_5
 \]
 Buyers are labelled with any ascending order of $d(\cdot)$ satisfying that buyer $b_{m+1}$ has priority over buyer $c_{m+1}$.
 %
 Finally, the number of unit is $k= m+2$ and the threshold value is $K = \epsilon + m\cdot v_1 + v_4$.

 Now, we show that there exists a set $\boldsymbol{A}' \subseteq
 \boldsymbol{A}$ such that $\sum_{i \in \boldsymbol{A}'} v(i)$ if and only if
 there exists a set $r'_s \subseteq r_s $ such that $\sum_{i \in N} t_i(\boldsymbol{\theta'} \mid r'_s) \geq \epsilon + m\cdot v_1 + v_4$.

 ($\Rightarrow$) Assume that there exists a subset $\boldsymbol{A}'
 \subseteq \boldsymbol{A}$ such that $\sum_{i \in A'} v(i) = m.$
 Consider the set ${r'_s = \{a^i_0 \in r_s: i \in \boldsymbol{A}'\} \cup
 \{b_1\} \cup \{c_1\}}$.
 Notice that when the seller shares the information with $r'_s$, the number of buyers with value $v_1$ involved in the auction is exactly $m$.
 Then the mechanism runs as follows:
 \begin{itemize}
 	\item Buyers $(a^i_0)_{i \in \boldsymbol{A}}$ receive price $v_3$, and
 	thus they do not buy.
 	\item For buyer $b_1$, there are $(m+1)$ values that are greater than $v_2$ (one value $v_5$ from buyer $c_{m+1}$ and $m$ values $v_1$ from buyers $(a^i_j)_{a^i_0 \in r'_s, j \in [1, \ldots, v(i)]}$).
 		Hence, buyer $b_1$ receives price $\epsilon$ and buys one unit.
 	\item Buyer $c_1$ receives price $v_1$ and does not buy.
 	\item Buyers $(a^i_j)_{a^i_0 \in r'_s, j\in [1, \ldots, v(i)]}$ receive price $v_3$ and do not buy.
 	\item For $j \in [2, \ldots, m]$:
 		\begin{itemize}
 			\item Buyer $b_{j}$ receives price $v_1$ and buys one unit.
 			\item Buyer $c_{j}$ receives price $v_3$ and does not buy.
 		\end{itemize}
 	\item For $j = m+1$, there are two units remaining, and thus buyer $b_{m+1}$ receives price $v_1$ and buys an unit.
 		Then, buyer $c_{m+1}$ receives price $v_4$ and buys an unit.
 		Finally, there is no unit left for buyer $b_{m+2}$.
 \end{itemize}
 Therefore, the revenue is $\epsilon + m\cdot v_1 + v_4$.

 ($\Leftarrow$) Assume that there exists $r'_s \subseteq r_s$ such that
 ${\sum_{i \in N} t_i(\boldsymbol{\theta'} \mid r'_s) \geq \epsilon +
 m\cdot v_1 + v_4}$.
 Since $v_3 << v_4$ holds, the threshold value can only be attained when buyer $c_{m+1}$ buys an unit at price $v_4$.
 Since there exists only one buyer with value $v_4$ (buyer $b_{m+2}$), it implies that $b_1$ and $c_1$ belong to $r'_s$, and that there is only one unit left when the mechanism reaches buyer $c_{m+1}$.
 Let us denote $\#v_1$ the number of buyers with value $v_1$ that are followers of buyers in $r'_s$.
 Note that, independently from the value taken by $\#v_1$, the buyers
 $(a^i_0)_{i \in \boldsymbol{A}}$ always receive price $v_3$ and thus they
 never buy.
 There are three cases to consider.

 First, assume that $\#v_1 >m$, then the mechanism runs as follows:
 \begin{itemize}
 	\item For buyer $b_1$, there are at least $m+2$ values that are greater than $v_2$ (one value $v_5$ from buyer $c_{m+1}$ and at least $m+1$ values $v_1$).
 		Hence, buyer $b_1$ receives price $v_1$ and does not buy.
 	\item Buyer $c_1$ also receives price $v_1$ and does not buy.
 	\item Buyers $(a^i_j)_{a^i_0 \in r'_s, j\in [1, \ldots, v(i)]}$ receive price $v_3$ and do not buy.
 	\item For $j \in [2, \ldots, m]$:
 		\begin{itemize}
 			\item Buyer $b_{j}$ receives price $v_1$ and buys one unit.
 			\item Buyer $c_{j}$ receives price $v_1$ and does not buy.
 		\end{itemize}
 	\item For $j = m+1$, there are three units remaining, and thus buyer $b_{m+1}$ receives price $v_1$ and buys.
 		Then, buyer $c_{m+1}$ also receives price $v_1$ and buys.
 		Finally, there is one unit left for buyer $b_{m+2}$ with price $v_1$ and he buys.
 \end{itemize}
 Hence the revenue is $(m+2)\cdot v_1$, which is a contradiction.

 Now, assume that $\#v_1 <m$, then the mechanism runs as follows:
 \begin{itemize}
 	\item For buyer $b_1$, there are at most $m$ values that are greater than $v_2$ (one value $v_5$ from buyer $c_{m+1}$ and at most $m-1$ values $v_1$).
 		Hence, buyer $b_1$ receives price $\epsilon$ and buys an unit.
 	\item Buyer $c_1$ receives price $v_3$ and does not buy.
 	\item Buyers $(a^i_j)_{a^i_0 \in r'_s, j\in [1, \ldots, v(i)]}$ receive price $v_3$ and do not buy.
 	\item For $j \in [2, \ldots, m]$:
 		\begin{itemize}
 			\item Buyer $b_{j}$ receives price \[
 		\begin{cases}
 			\epsilon,& \text{if } j \leq m-\#v+1,\\
 			v_1,& \text{if } j > m-\#v+1.
 		\end{cases}
 		\] and buys one unit.
 			\item Buyer $c_{j}$ receives price $v_3$ and does not buy.
 		\end{itemize}
 	\item For $j = m+1$, there are two units remaining.
 		Buyer $b_{m+1}$ and then buyer $c_{m+1}$ receives price \[
 		\begin{cases}
 			\epsilon,& \text{if } \#v=0,\\
 			v_1,& \text{if } \#v\geq 1.
 		\end{cases}
 		\] and buys one unit each.
 		Finally, there is no unit left for buyer $b_{m+2}$.
 \end{itemize}
 Hence the revenue is \[
 		\begin{cases}
 			\epsilon + (m-\#v_1)\cdot \epsilon + (\#v_1+1) \cdot v_1 & \text{if } \#v\neq 0,\\
 			(m+2) \cdot\epsilon & \text{if } \#v= 0.
 		\end{cases}
 		\]
 		Thus, it is strictly lower than $K$, which is a contradiction.

 Finally, assume that $\#v_1 =m$, then the mechanism runs as it is described in the argument for $(\Rightarrow)$, and the revenue is thus $\epsilon + m\cdot v_1 + v_4$.

 Therefore, in $r'_s$ and their followers, there exists exactly $m$ buyers with value $v_1$.
 Hence the set $\boldsymbol{A}' =\{ i \in \boldsymbol{A}: a^i_0 \in
 r'_s\}$ is such that $\sum_{i \in \boldsymbol{A}'} v(i) = m$.
\end{proof}

\end{document}